\documentclass[pra,
twocolumn,
superscriptaddress,
preprintnumbers,amsmath,amssymb]{revtex4}

\usepackage{graphicx}
\usepackage[caption=false]{subfig}
\usepackage{booktabs}
\usepackage{diagbox}
\usepackage{amsthm}
\usepackage{tensor}
\usepackage{color}
\usepackage[all]{xy}
\usepackage{tikz}
\usepackage{dsfont}
\usepackage{times,txfonts}
\usetikzlibrary{positioning}
\usepackage{braket}
\usepackage{mathtools}
\usepackage[T1]{fontenc}
\usepackage{enumitem}
\usepackage{xspace}
\usepackage{bm}
\usepackage{balance}
\usepackage{hyperref}

\setlist[enumerate]{label=(\roman*),align=left,labelsep=.5em,leftmargin=*,widest=ix,itemsep=0pt,parsep=0pt,topsep=.3\baselineskip}

\newtheorem{Thm}{Theorem}
\newtheorem{Lem}[Thm]{Lemma}
\newtheorem{Prop}[Thm]{Proposition}
\newtheorem{Cor}[Thm]{Corollary}
\theoremstyle{definition}
\newtheorem{Def}[Thm]{Definition}
\newtheorem{Rem}[Thm]{Remark}
\newtheorem{Example}[Thm]{Example}

\newcommand{\rank}{\operatorname{rank}}
\newcommand{\Tr}{\operatorname{Tr}}

\newcommand{\M}{\mathbb{M}}
\newcommand{\SN}{\mathrm{SN}}
\newcommand{\SR}{{\mathrm{SR}}}
\newcommand{\APPT}{\mathrm{APPT}}
\newcommand{\ARED}{\mathrm{ARED}}
\newcommand{\RED}{\mathrm{RED}}

\newcommand{\ASN}{\mathrm{ASN}}

\newcommand{\BP}{\mathrm{BP}}
\newcommand{\LS}{\mathrm{LS}}
\newcommand{\C}{\mathbb{C}}

\newcommand{\spec}{\operatorname{spec}}

\newcommand{\ox}{\otimes}

\newcommand{\vket}[1]{|#1\rangle\!\rangle}
\newcommand{\vbra}[1]{\langle\!\langle #1|}
\newcommand{\norm}[1]{\left\lVert #1\right\rVert}

\newcommand{\proj}[1]{\ket{#1}\!\bra{#1}}

\begin{document}

\title{A threshold for maximal Schmidt number from spectrum}
\author{Junhyeong An}
\email{hsj0356@khu.ac.kr}
\affiliation{Department of Mathematics and Research Institute for Basic Sciences, Kyung Hee University, Seoul 02447, Korea}
\author{Soojoon Lee}
\email{level@khu.ac.kr}
\affiliation{Department of Mathematics and Research Institute for Basic Sciences, Kyung Hee University, Seoul 02447, Korea}

\date{\today}

\begin{abstract}
The Schmidt number quantifies the dimensionality of entanglement in bipartite quantum states.
We investigate when 
the spectrum of a state on $\C^d \otimes \C^d$ alone guarantees that its Schmidt number is not maximal.
By deriving spectral bounds for $(d-1)$-block positive operators, we prove that $\LS_p\subseteq \ASN_{d-1}$ for $p\leq \lfloor d^2/2\rfloor$, where $\LS_p$ denotes the set of states whose largest eigenvalue does not exceed the sum of their $p$ smallest eigenvalues, and $\ASN_{d-1}$ denotes the set of states whose Schmidt number remains at most $d-1$ under all global unitaries.
This also yields a simple sufficient condition involving only the largest eigenvalue.
As an application, we show that states satisfying the reduction criterion under arbitrary global unitaries belong to $\ASN_{d-1}$ in every dimension. This subsumes the earlier result that absolutely positive partial transpose states belong to $\ASN_{d-1}$. The global unitary requirement is nevertheless essential:
for every $d\geq3$, we construct states of maximal Schmidt number satisfying the reduction criterion on both subsystems.
\end{abstract}

\maketitle

\section{Introduction}

Entanglement is one of the central resources in quantum information theory, and therefore not only detecting entanglement but also characterizing its structure is a fundamental problem~\cite{HorodeckiEtAl2009,GuhneToth2009}. 
An important aspect of this structure is the dimensionality of entanglement, which is naturally quantified by the Schmidt number: the smallest integer $k$ such that a bipartite state can be written as a convex combination of pure states of Schmidt rank at most $k$~\cite{TerhalHorodecki2000,SanperaBrussLewenstein2001}.
A larger Schmidt number can provide genuine advantages in several quantum information tasks, including channel discrimination and high-dimensional quantum communication~\cite{BaeChruscinskiPiani2019,CozzolinoDaLioBaccoOxenlowe2019,ErhardKrennZeilinger2020}.

For a bipartite system $\C^d\otimes\C^d$, let $\SN_k$ denote the convex set of states with Schmidt number at most $k$.
Since $\SN_d$ is the whole state space, the last nontrivial set is $\SN_{d-1}$, and every state outside $\SN_{d-1}$ has maximal Schmidt number $d$.
Thus, understanding maximal Schmidt number is equivalent to understanding 
the set $\SN_{d-1}$.

Since $\SN_{d-1}$ is convex, the Hahn--Banach separation theorem allows every state outside this set to be separated by a hyperplane.
In entanglement theory, such separating Hermitian operators are called Schmidt number witnesses~\cite{TerhalHorodecki2000,SanperaBrussLewenstein2001,ChruscinskiSarbicki2014}.
Accordingly, the dual objects associated with $\SN_{d-1}$ are the $(d-1)$-block positive operators denoted as $\BP_{d-1}$.

However, a complete characterization of $\BP_{d-1}$ is difficult, so directly determining all relevant Schmidt number witnesses is also difficult~\cite{Sarbicki2008,JohnstonKribs2010,JohnstonPatterson2018}.
On the other hand, spectral information is much simpler to handle and is independent of the choice of eigenbasis.
This motivates us to ask whether the spectral information of a state alone can guarantee that its Schmidt number is not maximal.

The spectrum of a state does not determine its entanglement, since global unitary transformations preserve the spectrum while changing the eigenvectors and, in general, the Schmidt number \cite{KusZyczkowski2001,IshizakaHiroshima2000,VerstraeteEtAl2001}.
Recently, the notion of Schmidt number from spectrum has been the subject of active research.
A state $\rho$ belongs to $\ASN_k$ if $\SN(U\rho U^\dagger)\le k$ for every global unitary $U$~\cite{AbellanetVidalEtAl2026,MallickEtAl2026}.
Thus, $\ASN_k$ consists precisely of states whose spectra cannot generate Schmidt number larger than $k$ under any change of eigenbasis.

To study such spectral conditions, we consider the family $\LS_p$ introduced in the study of positive reduction from spectrum~\cite{Jivulescu2015}.
$\LS_p$ consists of spectra for which the largest eigenvalue is no larger than the sum of the $p$ smallest eigenvalues, and these sets form an increasing hierarchy as $p$ increases.
We here prove that $\LS_p\subseteq \ASN_{d-1}$ whenever $p\le \lfloor d^2/2\rfloor.$

As a direct consequence, we also obtain a simple sufficient condition involving only the largest eigenvalue.
This provides an easily computable criterion: the largest eigenvalue of a state alone suffices to guarantee that its Schmidt number is not maximal, without requiring the full spectrum.

Among entanglement criteria arising from positive maps, two representative ones are the positive partial transpose (PPT) criterion and the reduction (RED) criterion~\cite{Peres1996,HorodeckiEtAl1996,HorodeckiReduction1999,cerf1999reduction}. 
Both provide simple necessary conditions for separability, and the reduction criterion is known to be weaker than the PPT criterion~\cite{HorodeckiReduction1999,cerf1999reduction}.
Their absolute versions are defined analogously to $\ASN_k$: APPT consists of states whose every global unitary image is PPT, while ARED is defined using the reduction criterion~\cite{Hildebrand2007PPT,Jivulescu2015}. 
It is known that $\APPT \subseteq \LS_3$ and $\ARED \subseteq \LS_{2d-1}$~\cite{Jivulescu2015}. 
Since $\lfloor d^2/2 \rfloor$ grows quadratically in $d$, whereas these known bounds are constant and linear in $d$, the range of $p$ covered by our result extends far beyond them as the local dimension grows.


It is already known that APPT states cannot have maximal Schmidt number~\cite{HuberLamiLancienMullerHermes2018}.
Since $\mathrm{ARED}\subseteq\mathrm{LS}_{2d-1}$, our $\mathrm{LS}_p$ threshold, together with an additional spectral estimate in the two-qutrit case, extends this conclusion to absolute reduction positivity and yields $\ARED\subseteq\ASN_{d-1}.$
This raises the stronger question of whether the reduction criterion alone, without the absoluteness requirement, guarantees that the Schmidt number of a state is not maximal.

For PPT states, whether $\mathrm{PPT}\subseteq\SN_{d-1}$ holds for general $d$ remains an important open question, although PPT states with large Schmidt number exist~\cite{HuberLamiLancienMullerHermes2018,johnston2026ppt}.
Since the reduction criterion is weaker than the PPT criterion, 
it is natural to
expect a negative answer in the reduction case, and we confirm this: for every $d \ge 3$, we construct states satisfying the reduction criterion on both subsystems that nevertheless have maximal Schmidt number. Hence
the absoluteness requirement is essential for turning the reduction criterion into a spectral obstruction to maximal Schmidt number.

The remainder of this paper is organized as follows.
In Sec.~\ref{sec:preli},  we review the basic definitions and introduce the mathematical tools used throughout the paper.
In Sec.~\ref{sec:witness}, we derive two complementary spectral estimates for $(d-1)$-block positive operators: a partial sum bound and a product bound for paired positive eigenvalues.
In Sec.~\ref{sec:SNFS}, we use the partial sum estimate to establish a general spectral criterion excluding maximal Schmidt number.
In Sec.~\ref{sec:application}, as an application, we examine whether the reduction criterion is compatible with maximal Schmidt number. 
We show that absolute reduction positivity excludes maximal Schmidt number in every dimension, whereas ordinary reduction positivity, even on both subsystems, does not.
Finally, in Sec.~\ref{sec:discussion},  we summarize the implications of our results and discuss remaining questions.

\section{Preliminaries}\label{sec:preli}

We first recall the notations and basic notions used throughout the paper.
Let $\M_d$ be the set of $d \times d$ complex matrices.
Then we can identify $\C^d\otimes\C^d$ with $\M_d$ by vectorization,
\[
X=(x_{ij})\in \M_d
\longleftrightarrow
\vket{X}
:=
\sum_{i,j}x_{ij}\ket{ij}\in\C^d\otimes \C^d,
\]
so that $\SR(\vket{X})=\rank X.$
The Schmidt number of a bipartite state $\rho\in \M_d\otimes \M_d$ is defined by~\cite{TerhalHorodecki2000}
\[
\SN(\rho)
=
\min_{\rho=\sum_i \proj{\psi_i}}
\max_i \SR(\ket{\psi_i}).
\]

For $1\le k\le d$, let $\SN_k$ be the closed convex cone generated by pure states of Schmidt rank at most $k$.
To detect whether a state lies outside $\SN_k$, one uses $k$-block positive operators, which form the dual cone of $\SN_k$:
\[
\begin{aligned}
\BP_k
=
\{
W: W=W^\dagger,
\langle\psi|W|\psi\rangle\ge0\text{ for all }\ket{\psi}\bra{\psi}\in\SN_k\}.
\end{aligned}
\]
Consequently,
\[
\SN(\rho)\le k
\Longleftrightarrow
\Tr(W\rho)\ge0,
\text{ }\forall W\in\BP_k.
\]
In particular, states of maximal Schmidt number are detected by $(d-1)$-block positive operators.

Our aim is to obtain conditions on the spectrum of a state guaranteeing that its Schmidt number is not maximal. 
Thus, we need the following definition~\cite{MallickEtAl2026,AbellanetVidalEtAl2026}.

\begin{Def}[Schmidt number from spectrum]
A state $\rho\in\M_d\otimes\M_d$ belongs to $\ASN_k$ if
\[
\SN(U\rho U^\dagger)\le k
\]
for every unitary $U$ on $\C^d\otimes\C^d$.
\end{Def}

Ref.~\cite{AbellanetVidalEtAl2026} provides the following spectral characterization of $\ASN_k$:
\begin{equation}
\rho\in\ASN_k
\Longleftrightarrow
\sum_{i=1}^{d^2}\lambda_i^\downarrow(\rho)\mu_i^\uparrow(W)\ge0,
\forall W\in \BP_k
\label{eq:ASN_k}    
\end{equation}

where $\lambda_i^\downarrow(\rho)$ is the $i$-th eigenvalue of $\rho$ arranged in decreasing order and $\mu_i^\uparrow(W)$ is the $i$-th eigenvalue of $W$ arranged in increasing order.

To investigate $\ASN_{d-1}$, we use the classes $\LS_p$ introduced in Ref.~\cite{Jivulescu2015}, which are defined as
\[
\rho\in\LS_p\quad\Longleftrightarrow\quad
\lambda_1^\downarrow(\rho)\le
\sum_{i=d^2-p+1}^{d^2}\lambda_i^\downarrow(\rho).
\]

We next recall the reduction map $R(X) = \mathrm{Tr}(X) I_d - X$.
We write
$\rho \in \mathrm{RED}_A$ if and only if $(R \otimes \mathrm{id})(\rho) \ge 0$, and $\rho \in \mathrm{RED}_B$ if and only if $(\mathrm{id} \otimes R)(\rho) \ge 0$,
where $(\mathrm{id} \otimes R)(\rho) = \rho_A \otimes I_d - \rho$ and $(R \otimes \mathrm{id})(\rho) = I_d \otimes \rho_B - \rho$.

Let $S$ be the swap operator.
Since $S(\mathrm{id} \otimes R)(X)S^\dagger = (R \otimes \mathrm{id})(SXS^\dagger)$, we have $\rho \in \mathrm{RED}_B$ if and only if $S\rho S^\dagger \in \mathrm{RED}_A$. 
As $S$ is a global unitary, the two reduction criteria from spectrum obtained from $R \otimes \mathrm{id}$ and $\mathrm{id}\otimes R$ coincide, so for absolute reduction positivity, it suffices to consider one side as in the following definition~\cite{Jivulescu2015}.

\begin{Def}[Positive reduction from spectrum]
A state $\rho\in\M_d\otimes\M_d$ belongs to $\ARED$ if
\[
(\operatorname{id}\otimes R)(U\rho U^\dagger)\ge0
\]
for every global unitary $U$.
\end{Def}

Moreover, Ref~\cite{Jivulescu2015} provides the spectral characterization of ARED.
For any probability vector $x=(x_1,x_2,\ldots, x_d)$ define $\ket{\psi_x}=\sum_i\sqrt{x_i}\ket{ii}$ and  $W_x=(\operatorname{id}\ox R)\ket{\psi_x}\bra{\psi_x}$.
Then 
\[
\rho\in\ARED
\Longleftrightarrow
\sum_{i=1}^{d^2}\lambda_i^\downarrow(\rho)\mu_i^\uparrow(W_x)
\ge0, \forall x.
\]

\section{Spectral properties of $(d-1)$-block positive operators}\label{sec:witness}

We now study the spectral properties of $(d-1)$-block positive operators, which are the witnesses relevant to maximal Schmidt number.
Let $W\in\BP_{d-1}$ and assume $W\ngeq0$.
It is known that a $k$-block positive operator on $\C^d \otimes \C^d$ has at most $(d-k)^2$ negative eigenvalues~\cite{Sarbicki2008,JohnstonKribs2010}; hence $W$ has exactly one negative eigenvalue, and the corresponding eigenvector clearly has Schmidt rank $d$.

For our spectral analysis, we also need to exclude zero eigenvalues.
The following theorem shows that every non-positive $(d-1)$-block positive operator is necessarily invertible.

\begin{Thm}\label{thm:bp-inertia}
Let $W\in\BP_{d-1}$ and assume $W\ngeq0$.
Then $W$ has no zero eigenvalue.
Equivalently, $W$ is invertible.
\end{Thm} 
\noindent
The proof of Theorem~\ref{thm:bp-inertia} is given in Appendix~\ref{app:bp-inertia}.

Let $W\in\BP_{d-1}$ with $W\ngeq0.$
By Theorem~\ref{thm:bp-inertia}, it has exactly one negative eigenvalue $-a<0$, no zero eigenvalue, and its normalized negative eigenvector may be written as $\vket{A}$ with $A$ invertible.
Thus, we may write
\begin{equation}\label{eq:witness-spectrum}
\mu^\uparrow(W)
=
(-a,b_1,\dots,b_{d^2-1})\in \mathbb R^{d^2},
\quad b_{i+1}\ge b_i>0
\end{equation}
and set $B_r:=\sum_{i=1}^r b_i.$
In this paper, we refer to $B_r$ as the partial sums of the positive eigenvalues.

\subsection{Partial sum bounds}

With $a$ and $A$ as above, set $G=W+a\vket{A}\vbra{A}.$ 
Then
\begin{equation}\label{eq:negative-direction-setup}
\begin{gathered}
G\ge0,
\;
G\vket{A}=0,\\
\;
\vbra{A}A\rangle\!\rangle=1,
\;
\spec(G)=(0,b_1,\ldots,b_{d^2-1}).
\end{gathered}
\end{equation}


Let $A^\perp := \{\, Y \in \M_d : \mathrm{Tr}(A^\dagger Y) = 0 \,\}$ denote the orthogonal complement of $A$ in $\M_d$ with respect to the Hilbert–Schmidt inner product.
For $X\in\M_d$, let
\[
\operatorname{spr}(X)
:=
\max\{|z|:z\in\spec(X)\}
\]
be its spectral radius, where $\spec(X)$ is the set of eigenvalues of $X$.
Then we obtain the following lemma.

\begin{Lem}\label{lem:det-energy}
Let $W\in\BP_{d-1}$ with $W\ngeq0$, and let $G$, $a$, and $A$ be as in Eq.~\eqref{eq:negative-direction-setup}. 
For every $Y\in A^\perp$,
\begin{equation}\label{eq:rho-bound}
\vbra{Y}G\vket{Y} \ge a\operatorname{spr}(A^{-1}Y)^2.
\end{equation}
Consequently,
\begin{equation}\label{eq:quadratic-bound}
\frac1a\vbra{Y}G\vket{Y}\ge\frac1d|\Tr[(A^{-1}Y)^2]|.
\end{equation}
\end{Lem}

\begin{proof}
Let $\zeta$ be an eigenvalue of $A^{-1}Y$. 
Then $X=Y-\zeta A=A(A^{-1}Y-\zeta I)$ is singular, so $\rank X\le d-1$. $(d-1)$-block positivity gives
\[
0\le\vbra{X}W\vket{X}=\vbra{Y}G\vket{Y}-a|\zeta|^2,
\]
because $G\vket{A}=0$.
Choosing $\zeta$ with $|\zeta| = \operatorname{spr}(A^{-1}Y)$ proves the inequality in Eq.~\eqref{eq:rho-bound}. 
If $\zeta_1,\ldots,\zeta_d$ are the eigenvalues of $A^{-1}Y$, then $|\Tr[(A^{-1}Y)^2]|=|\sum_i\zeta_i^2|\le \sum_i|\zeta_i|^2 \le d\operatorname{spr}(A^{-1}Y)^2$.
\end{proof}

The following lemma guarantees the existence of a nonzero normal matrix in a sufficiently large subspace.




\begin{Lem}\label{lem:normal-subspace}
Let $V\subset \M_d$ be a complex linear subspace. 
If
\[
\dim_{\mathbb C}V\ge \left\lceil \frac{d^2}{2}\right\rceil,
\]
then $V$ contains a nonzero normal matrix. 
More precisely, $V$ contains a nonzero matrix of the form
\[
Z=H+cI_d,
\]
where $H=H^\dagger$ and $c\in\C$.
\end{Lem}
\noindent
A short proof of Lemma~\ref{lem:normal-subspace} is given in Appendix~\ref{app:normal}.

\begin{Prop}\label{prop:Br}
Let $W\in\BP_{d-1}$ with $W\ngeq0$ and write its spectrum as in Eq.~\eqref{eq:witness-spectrum}.
For any integer \(r\) satisfying $\left\lceil\frac{d^2}{2}\right\rceil\le r\le d^2-1$,
\begin{equation}\label{eq:Br}
B_r\ge\left(r-\left\lceil \frac{d^2}{2}\right\rceil+1\right)a.
\end{equation}
\end{Prop}

\begin{proof}
Let $E\subset A^\perp$ be an arbitrary $r$-dimensional complex subspace.  
Since $A^{-1}$ is invertible, $\dim_{\C}(A^{-1}E)=r,$ where $A^{-1}E=\{A^{-1}Y:Y\in E\}.$
By Lemma~\ref{lem:normal-subspace}, $A^{-1}E$ contains a nonzero normal matrix $Z_1$.  
Without loss of generality we may assume  
\[
Y_1=AZ_1,
\text{ }
\vbra{Y_1}Y_1\rangle\!\rangle=1,
\text{ }
Z_1\text{ is normal}.
\]
Set $E_1=E\cap Y_1^\perp.$
As long as the remaining dimension is at least $\lceil d^2/2\rceil$, the same argument may be repeated. 
We obtain orthonormal vectors $Y_1,\dots,Y_{r-\lceil d^2/2\rceil+1}\in E$ such that $ Y_j=AZ_j$, and $Z_j$ is normal.

Let $s_1(X)$ denote the largest singular value of $X$. 
Recall that $\operatorname{spr}(X)\le s_1(X)$ with equality when \(X\) is normal~\cite{Bhatia1997,HornJohnson2013}.
Since each \(Z_j\) is normal, $s_1(Z_j)=\operatorname{spr}(Z_j)$.
For each $j$,
\[
\begin{aligned}
1&=\Tr(A^\dagger A Z_jZ_j^\dagger)
\le s_1(Z_j)^2\Tr(A^\dagger A)\\
&=s_1(Z_j)^2
=\operatorname{spr}(Z_j)^2.
\end{aligned}
\]
Applying Lemma~\ref{lem:det-energy} to $Y_j=AZ_j$,
\[
\vbra{Y_j}G\vket{Y_j}
\ge
a\operatorname{spr}(Z_j)^2
\ge a.
\]

Extending $\{\vket{Y_i}\}$ to an orthonormal basis of $E$, let $P_E=\sum_{i=1}^{r}\vket{Y_i}\vbra{Y_i}$ be the orthogonal projection onto $E$.
Then using $G\ge0$, we obtain $\Tr(P_EG)\ge(r-\lceil d^2/2\rceil+1)a$. 
Ky Fan's minimum principle~\cite{Zhan2013} yields 
\[
B_r
=
\min_{\substack{\dim E=r \\ E\subseteq A^\perp}}\Tr(P_EG)\ge\left(r-\left\lceil \frac{d^2}{2}\right\rceil+1\right)a.
\]
\end{proof}

\subsection{Product bounds for paired positive eigenvalues}

Set $Q:=\frac1aG|_{A^\perp}$.
Define the complex symmetric bilinear form
\begin{equation}\label{eq:takagi-setup}
\beta(Y,Z):=\frac1d\Tr(A^{-1}YA^{-1}Z),
\quad Y,Z\in A^\perp.
\end{equation}
With this notation, $\vbra{Y}Q\vket{Y}\ge|\beta(Y,Y)|$ by Eq.~(\ref{eq:quadratic-bound}).

Let $\mathcal F=\{F_1,\ldots,F_{d^2-1}\}$ be a Hilbert--Schmidt orthonormal basis of $A^\perp$ and define $
S_\beta^{\mathcal F}
=
[\beta(F_i,F_j)]_{i,j=1}^{d^2-1}.$
Since $\beta$ is symmetric, $S_\beta^{\mathcal F}$ is complex symmetric. 
Although the coordinate matrix $S_\beta^{\mathcal F}$ depends on the choice of the orthonormal basis $\mathcal F$, its singular values do not. 
Changing the orthonormal basis acts by unitary congruence, so the singular values of this matrix are basis independent; see Appendix~\ref{appen:S_beta}.
Accordingly, whenever only its singular values are relevant, we simply write $S_\beta$.

\begin{Lem}\label{lem:takagi-lower}
Let $S_\beta$ be the matrix representation, with respect to an orthonormal basis of $A^\perp$, of the bilinear form $\beta$ defined in Eq.~\eqref{eq:takagi-setup}.
Then every singular value of $S_\beta$ is at least $2/d$.
\end{Lem}

\begin{proof}
By Appendix~\ref{appen:S_beta}, we may assume $A=\operatorname{diag}(s_1,\ldots,s_d)$ with $s_i>0$ and $\sum_i s_i^2=1$. 
Write $A^\perp=\mathcal O\oplus\mathcal D$, where $\mathcal O=\operatorname{span}\{\ket{i}\bra{j}:i\ne j\}$ and $\mathcal D=\{\operatorname{diag}(z_1,\dots,z_d):\sum_is_iz_i=0\}$.
The two subspaces are orthogonal with respect to both the Hilbert--Schmidt inner product and $\beta$, that is $\beta(\ket{i}\bra{j},D)=0$ for all $D\in\mathcal{D}$.

On $\operatorname{span}\{\ket{i}\bra{j},\ket{j}\bra{i}\}$, $i<j$, the matrix of $\beta$ is
\[
\begin{pmatrix}
0 & \frac{1}{ds_is_j}\\
\frac{1}{ds_is_j} & 0
\end{pmatrix}.
\]
The singular values of off-diagonal blocks are at least $2/d$, because $2s_is_j\le s_i^2+s_j^2\le1$. Appendix~\ref{appen:pf-takagi-lower} proves the same lower bound on the diagonal part. 
Hence every singular value of $S_\beta$ is at least $2/d$.
\end{proof}

\begin{Cor}\label{cor:pair-product}
Let $W\in\BP_{d-1}$, $W\ngeq0$ with spectrum as in Eq.~\eqref{eq:witness-spectrum}.
Then
\begin{equation}\label{eq:pair-general}
b_i b_{d^2-i}\ge\frac{4a^2}{d^2},\; i=1,\ldots,\left\lfloor\frac{d^2}{2}\right\rfloor.
\end{equation}
\end{Cor}

\begin{proof}
For the given $W$, we here use the objects $A$, $G$, $Q$, $\beta$, and $S_\beta$ from Eqs.~\eqref{eq:negative-direction-setup} and \eqref{eq:takagi-setup}. Then the $(d^2-1)\times (d^2-1)$ matrix $Q$ is positive definite by Theorem~\ref{thm:bp-inertia}.

Define $K=(Q^{-1/2})^T S_\beta Q^{-1/2}$ and $\vket{Y}=Q^{1/2}\vket{X}$.
Then since $\vbra{X}Q\vket{X}\ge|\beta(X,X)|$, 
we have
\[
\vbra{Y}Y\rangle\!\rangle 
=
\vbra{X}Q\vket{X}
\ge
|\beta(X,X)|
=
|\vbra{X^*}S_\beta\vket{X}|
=
|\vbra{Y^*}K\vket{Y}|
\]
Since $K^T=K$, for any unit vectors $\vket{U}$ and $\vket{V}$,
\[
4\vbra{U^*}K\vket{V}
=
\vbra{(U+V)^*}K\vket{U+V}
-
\vbra{(U-V)^*}K\vket{U-V}.
\]
It follows that  
\[
\begin{aligned}
\left|
\vbra{U^*}K\vket{V}
\right|
&\le
\frac14
\left(
\norm{U+V}^2+\norm{U-V}^2
\right)
=1.
\end{aligned}
\]
For a fixed unit vector $\vket{Y}$, set $\vket{Z}=\vket{X^*}$.
Complex conjugation is a bijection of the unit sphere, and therefore the standard dual characterization of the Euclidean norm gives
\[
\norm{K\vket{Y}}
=
\sup_{\norm{Z}=1}
\left|
\vbra{Z}K\vket{Y}
\right|
=
\sup_{\norm{X}=1}
\left|
\vbra{X^*}K\vket{Y}
\right|
\le1.
\]
Hence $\norm{K}\le1$.

Let $\sigma_1(\cdot)\ge\cdots\ge \sigma_n(\cdot)$ denote the singular values. 
Using the inequality $\sigma_{i+j-1}(AB)\le \sigma_i(A)\sigma_j(B)$ with $j=d^2-i$~\cite{Bhatia1997,HornJohnson2013}, we obtain
\[
\begin{aligned}
\sigma_{d^2-1}(S_\beta)
&\le
\sigma_i((Q^{1/2})^T)
\sigma_{d^2-i}(KQ^{1/2})\\
&\le
\sigma_i((Q^{1/2})^T)
\sigma_{d^2-i}(Q^{1/2}),
\end{aligned}
\]
where the second inequality follows from $\norm{K}\le1$.

The singular values of both $Q^{1/2}$ and $(Q^{1/2})^T$ are  $\sqrt{b_1/a},\dots \sqrt{b_{d^2-1}/a}$.
Thus we have
\[
\sigma_i((Q^{1/2})^T)=\sqrt{\frac{b_{d^2-i}}{a}},
\;
\sigma_{d^2-i}(Q^{1/2})=\sqrt{\frac{b_i}{a}}.
\]
By Lemma~\ref{lem:takagi-lower},
\[
\frac{2}{d}
\le
\sigma_{d^2-1}(S_\beta)
\le
\sqrt{b_i b_{d^2-i}}/a.
\]
Hence we obtain 
\[
b_i b_{d^2-i}\ge\frac{4a^2}{d^2},
\]
for $i=1,\dots,\lfloor d^2/2\rfloor$.
\end{proof}

\section{A spectral threshold for maximal Schmidt number}\label{sec:SNFS}

We now turn the partial sum bound into a spectral criterion on states. 
The following lemma isolates exactly the partial sum inequalities required by the $\LS_p$ condition.

\begin{Lem}\label{lem:partial-sum}
For $W\in\BP_{d-1}$ with the spectrum as in Eq.~\eqref{eq:witness-spectrum},
\begin{equation}\label{eq:required-partials}
B_{d^2-p+r}\ge (r+1)a,\;\text{for }r=0,\dots,p-1,            
\end{equation}
if and only if 
$\sum_{i=1}^{d^2}\lambda_i^\downarrow(\rho)\mu_i^\uparrow(W)\ge0$
for every  $\rho\in\LS_p$.
\end{Lem}

\begin{proof}
Let $q=d^2-p$.
Suppose Eq.~(\ref{eq:required-partials}) holds.
For a given state $\rho\in \LS_p$, let $\lambda^\downarrow (\rho)= (\lambda_1,\dots,\lambda_{d^2})$ and write $y_j=\lambda_j-\lambda_{j+1}\ge0$, with $\lambda_{d^2+1}=0$. 
Since $\lambda_i=\sum_{j=i}^{d^2}y_j$, we note that
\[
\begin{aligned}
\sum_{i=q+1}^{d^2} \lambda_i
&=
\sum_{i=q+1}^{d^2}\sum_{j=i}^{d^2} y_j\\
&=
\sum_{j=q+1}^{d^2}
\#\{i:q+1\le i\le j\}\,y_j\\
&=
\sum_{j=q+1}^{d^2} (j-q)y_j.
\end{aligned}
\]
Setting $r=j-q-1$ rewrites the $\LS_p$ condition as
\[
\lambda_1\le \sum_{i=q+1}^{d^2}\lambda_i
\Longleftrightarrow
\sum_{j=1}^{q}y_j
\le
\sum_{r=0}^{p-1}r\,y_{q+1+r}.
\]
Let $M_j=\sum_{i=1}^j\mu_i$ for $\mu^\uparrow(W)=(\mu_1,\ldots,\mu_{d^2})$. 
Abel summation gives $\sum_{i=1}^{d^2}\lambda_i\mu_i=\sum_{j=1}^{d^2}y_jM_j$. 
Since $b_i>0$, $M_j\ge-a$ for $j\le q$. Eq.~\eqref{eq:required-partials} implies $M_{q+1+r}\ge ra$. Hence
\[
\sum_{i}\lambda_i\mu_i
= \sum_iy_iM_i
\ge -a\sum_{j=1}^{q}y_j+a\sum_{r=0}^{p-1}r\,y_{q+1+r}\ge0.
\]
Conversely, suppose that
$\sum_{i=1}^{d^2}\lambda_i^\downarrow(\rho)\mu_i^\uparrow(W)\ge0$
for every  $\rho\in\LS_p$.
Let $r$ be fixed, and choose $\rho$ whose $q+r+1$ nonzero eigenvalues are $\lambda_1=(r+1)t$, $\lambda_2=\dots=\lambda_{q+r+1}=t>0$ and the others are zero.
After normalization $\rho\in \LS_p$.
If $B_{q+r}<(r+1)a$ for some $r$, then
\[
\sum_{i=1}^{d^2}\lambda_i^\downarrow(\rho)\mu_i^\uparrow(W)
=
(-(r+1)a + B_{q+r})t<0,
\]
which is impossible.
\end{proof}

\begin{Thm}\label{thm:main}
For every $d\ge2$,
\[
\LS_p\subseteq\ASN_{d-1},
\
\text{whenever}
\
1\le p\le\left\lfloor\frac{d^2}{2}\right\rfloor.
\]
\end{Thm}

\begin{proof}
Let $\rho\in\LS_p$.
Then by Eq.~(\ref{eq:ASN_k}), it suffices to prove $\sum_{i=1}^{d^2}\lambda_i^\downarrow(\rho)\mu_i^\uparrow(W)\ge0$ for every $W\in\BP_{d-1}$.
For $W\ge0$, it is trivial, so we consider $W\ngeq0$.

Write $\mu^\uparrow(W)=(-a,b_1,\dots,b_{d^2-1})$. 
Since $p\le\lfloor d^2/2\rfloor$, one has $d^2-p\ge \lceil d^2/2\rceil$.
Applying Proposition~\ref{prop:Br},
\[
B_{d^2-p+r}\ge(d^2-p-\lceil d^2/2\rceil+r+1)a\ge(r+1)a.
\]
By Lemma~\ref{lem:partial-sum}, $\sum_{i=1}^{d^2}\lambda_i^\downarrow(\rho)\mu_i^\uparrow(W)\ge0$, and therefore $\rho\in\ASN_{d-1}$. 
\end{proof}

The $\LS_p$ condition itself requires information from several eigenvalues. 
A simpler, though coarser, sufficient criterion can be obtained from 
the largest eigenvalue alone.

\begin{Cor}\label{cor:lmax}
Let $\lambda^\downarrow(\rho)=(\lambda_1,\dots, \lambda_{d^2})$.
If
\begin{equation}\label{eq:lmax}
\lambda_{1}(\rho)
\le
\frac{1}{\lceil d^2/2\rceil+1},
\end{equation}
then $\rho\in\ASN_{d-1}$.
\end{Cor}

\begin{proof}

Suppose $\lambda_1(\rho)\le\frac{1}{\lceil d^2/2\rceil+1}$.
Then $1-\lceil d^2/2\rceil\lambda_1
\ge\lambda_1$.
So,
\[
\sum_{i=\lceil d^2/2\rceil+1}^{d^2}\lambda_i
=1-\sum_{i=1}^{\lceil d^2/2\rceil}\lambda_i
\ge1-\lceil d^2/2\rceil\lambda_1
\ge\lambda_1.
\]
Thus $\rho\in\LS_{\lfloor d^2/2\rfloor}$.
Therefore $\rho\in\ASN_{d-1}$ by Theorem~\ref{thm:main}.
\end{proof}


\section{Application to the reduction criterion}\label{sec:application}

We now apply Theorem~\ref{thm:main} to absolute reduction positivity.
The known inclusion $\ARED\subseteq\LS_{2d-1}$~\cite{Jivulescu2015} and the inequality $2d-1\le\lfloor d^2/2\rfloor$ imply $\ARED\subseteq\ASN_{d-1}$ immediately for $d\ge4$.

The case $d=3$ is different.
Since for $d=3$, $\lfloor d^2/2\rfloor =4$, Theorem~\ref{thm:main} implies $\LS_4\subseteq \ASN_2$ while only $\ARED \subseteq \LS_5$ has been known from~\cite{Jivulescu2015}.
Nevertheless, we can prove that $\ARED \subseteq \ASN_2.$

\subsection{The two-qutrit case}

Let $\lambda_1\ge\lambda_2\ge\cdots\ge\lambda_9\ge0$ be the eigenvalues of a two-qutrit state $\rho\in\ARED$.
Choose Schmidt probabilities $(u,v,0)$ with $u\ge v>0$ and $u+v=1$.
The increasing spectrum of $(\operatorname{id}\otimes R)\proj{\psi_{u,v}}$
is $(-\sqrt{uv},0,0,0,v,v,\sqrt{uv},u,u)$, where $\ket{\psi_{u,v}} =\sqrt{u}\ket{00} + \sqrt{v}\ket{11}$.
The ARED criterion with respect to $(\operatorname{id}\otimes R)\proj{\psi_{u,v}}$ gives 
\[-\sqrt{uv}\lambda_1+v(\lambda_5+\lambda_6)+\sqrt{uv}\lambda_7+u(\lambda_8+\lambda_9)\ge0.
\]

With $t=\sqrt{v/u}\in(0,1]$, division by $\sqrt{uv}$ yields
\[
\lambda_1
\le
t(\lambda_5+\lambda_6)
+\lambda_7
+\frac1t(\lambda_8+\lambda_9).
\]
Minimizing in $t$ yields
\begin{equation}\label{eq:qutrit-ared-min}
\lambda_1
\le
\lambda_7
+
2\sqrt{(\lambda_5+\lambda_6)(\lambda_8+\lambda_9)}.
\end{equation}
Thus every two-qutrit ARED state satisfies Eq.~(\ref{eq:qutrit-ared-min}).

On the other hand, by Corollary~\ref{cor:pair-product}, for every qutrit $2$-block positive witness with a negative eigenvalue,
\begin{equation}\label{eq:qutrit-pairs}
b_1b_8\ge\frac49a^2,\text{ }
b_2b_7\ge\frac49a^2,\text{ }
b_3b_6\ge\frac49a^2,\text{ }
b_4b_5\ge\frac49a^2.
\end{equation}
Equation~\eqref{eq:qutrit-ared-min} constrains the state spectrum, while Eq.~\eqref{eq:qutrit-pairs} constrains the dual witness spectrum.
The next lemma combines these two spectral constraints.

\begin{Lem}\label{lem:qutrit-pairing}
Let $\rho$ be a two-qutrit state whose decreasing spectrum satisfies Eq.~\eqref{eq:qutrit-ared-min}.
Then $\rho\in\ASN_2$.
\end{Lem}

\begin{proof}
This is trivial for $W\ge0$.
So assume that $W\ngeq 0$ and write $\mu^\uparrow(W)=(-a,b_1,\dots,b_8)$.
We first establish
\begin{equation}\label{eq:qutrit-pairing-bound}
\lambda_1\le\frac43\left(
\sqrt{\lambda_2\lambda_9}
+\sqrt{\lambda_3\lambda_8}
+\sqrt{\lambda_4\lambda_7}
+\sqrt{\lambda_5\lambda_6}
\right).
\end{equation}
Since $\lambda_2,\lambda_3,\lambda_4\ge\lambda_5$,
\[
\begin{aligned}
&\sqrt{\lambda_2\lambda_9}
+\sqrt{\lambda_3\lambda_8}
+\sqrt{\lambda_4\lambda_7}
+\sqrt{\lambda_5\lambda_6}\\
&\ge
\sqrt{\lambda_5}\left(
\sqrt{\lambda_6}+\sqrt{\lambda_7}+\sqrt{\lambda_8}+\sqrt{\lambda_9}
\right).
\end{aligned}
\]
Combining $\lambda_5\ge \lambda_6\ge\lambda_7$ with Eq.~\eqref{eq:qutrit-ared-min}, we obtain
\[
\lambda_1
\le
\sqrt{\lambda_5}\left(
\sqrt{\lambda_7}+2\sqrt2\sqrt{\lambda_8+\lambda_9}
\right).
\]
It therefore remains to verify
\[
\sqrt{\lambda_7}+2\sqrt2\sqrt{\lambda_8+\lambda_9}
\le
\frac43\left(
\sqrt{\lambda_6}+\sqrt{\lambda_7}+\sqrt{\lambda_8}+\sqrt{\lambda_9}
\right).
\]
Since $\lambda_6,\lambda_7\ge\lambda_8$, it is enough to show, for $\lambda_8>0$
\[
2\sqrt2\sqrt{\lambda_8+\lambda_9}
\le
3\sqrt{\lambda_8}+\frac43\sqrt{\lambda_9}.
\]
Put $\theta=\sqrt{\lambda_9/\lambda_8}\in[0,1].$
After division by $\sqrt{\lambda_8}$ and squaring, the inequality is equivalent to
\[
f(\theta)=1+8\theta-\frac{56}{9}\theta^2\ge0.
\]
$f$ is concave on $[0,1]$ and equals $1$ and $25/9$ at the endpoints.
Thus Eq.~\eqref{eq:qutrit-pairing-bound} follows.

Now pair the positive eigenvalues of $W$:
\[
\begin{aligned}
\sum_{i=1}^8b_i\lambda_{i+1}
={}&(b_1\lambda_2+b_8\lambda_9)
+(b_2\lambda_3+b_7\lambda_8)\\
&+(b_3\lambda_4+b_6\lambda_7)
+(b_4\lambda_5+b_5\lambda_6).
\end{aligned}
\]
By the arithmetic–geometric mean inequality and Eq.~\eqref{eq:qutrit-pairs},
\[
\begin{aligned}
\sum_{i=1}^8b_i\lambda_{i+1}
&\ge
\frac{4a}{3}\left(
\sqrt{\lambda_2\lambda_9}
+\sqrt{\lambda_3\lambda_8}
+\sqrt{\lambda_4\lambda_7}
+\sqrt{\lambda_5\lambda_6}
\right)\\
&\ge a\lambda_1,
\end{aligned}
\]
which implies $\sum_i\lambda_i^\downarrow\mu_i^\uparrow(W)\ge 0$.
Therefore $\rho\in\ASN_2$.
\end{proof}

\begin{Thm}\label{thm:ARED}
For every $d\ge2$,
\[
\ARED\subseteq\ASN_{d-1}.
\]
\end{Thm}

\begin{proof}
Let $\rho\in\ARED$.
For $d\ge4$, the result follows from $\ARED\subseteq\LS_{2d-1}$ and Theorem~\ref{thm:main}. 
For $d=3$, the spectrum of $\rho$ satisfies Eq.~\eqref{eq:qutrit-ared-min}. 
Hence $\rho\in \ASN_2$ by Lemma~\ref{lem:qutrit-pairing}.
Finally, for $d=2$, reduction positivity is equivalent to PPT, which is equivalent to separability~\cite{Peres1996,HorodeckiEtAl1996,HorodeckiReduction1999}. 
Hence $\rho \in \mathrm{ARED}$ is separable, that is, $\rho \in \mathrm{ASN}_1$.
\end{proof}

Similar to ARED, APPT is defined as the set of states $\rho$ satisfying $(\operatorname{id} \otimes T)(U\rho U^\dagger) \ge 0$ for all unitary $U$, where $T$ is the transpose map.
Since $\APPT\subseteq\ARED$, Theorem~\ref{thm:ARED} recovers the earlier exclusion of maximal Schmidt number for APPT states obtained in~\cite{HuberLamiLancienMullerHermes2018}, while extending it to the larger class ARED.

\begin{Rem}
In the qutrit setting, $\LS_4$ and the condition in
Eq.~\eqref{eq:qutrit-ared-min} are incomparable.
Let $\lambda^\downarrow(\rho)=(\lambda_1,\dots,\lambda_9)$ be the spectrum of $\rho$.
First, consider
\[
(\lambda_1,\dots,\lambda_9)
=
\frac{1}{10}(2,2,2,1,1,1,1,0,0).
\]
Then
\[
\lambda_1
=
\lambda_6+\lambda_7+\lambda_8+\lambda_9
=
\frac{1}{5},
\]
so $\rho$ belongs to $\LS_4$.
On the other hand,
\[
\lambda_7
+
2\sqrt{(\lambda_5+\lambda_6)(\lambda_8+\lambda_9)}
=
\frac{1}{10}
<
\lambda_1.
\]
Hence Eq.~\eqref{eq:qutrit-ared-min} does not hold.

Conversely, consider
\[
(\lambda_1,\dots,\lambda_9)
=
\frac{1}{21}(5,5,5,1,1,1,1,1,1).
\]
Then
\[
\lambda_6+\lambda_7+\lambda_8+\lambda_9
=
\frac{4}{21}
<
\frac{5}{21}
=
\lambda_1,
\]
so $\rho\notin\LS_4$. 
However,
\[
\lambda_7
+
2\sqrt{(\lambda_5+\lambda_6)(\lambda_8+\lambda_9)}
=
\frac{1}{21}
+
2\sqrt{\frac{2}{21}\frac{2}{21}}
=
\frac{5}{21}
=
\lambda_1.
\]
Thus Eq.~\eqref{eq:qutrit-ared-min} holds.
Hence Lemma~\ref{lem:qutrit-pairing} cannot be deduced from Theorem~\ref{thm:main} alone, and the product bounds in Eq.~\eqref{eq:qutrit-pairs} are genuinely needed in the qutrit case.
\end{Rem}

\subsection{Relation with the $\LS_p$ hierarchy}

It was shown in~\cite{Jivulescu2015} that, for $d \ge 3$,
\[
\APPT\subseteq\LS_3\subseteq\LS_d\subseteq\ARED\subseteq\LS_{2d-1}.
\]

For $d \ge 4$, we have $2d-1 \le \lfloor d^2/2 \rfloor$, so Theorem~\ref{thm:main} extends this hierarchy to
\begin{equation*}
\begin{aligned}
\APPT
&\subseteq\LS_3
\subseteq\LS_d
\subseteq\ARED
\subseteq\LS_{2d-1}\\
&\subseteq\LS_{\lfloor d^2/2\rfloor}
\subseteq\ASN_{d-1},
\end{aligned}
\end{equation*} 
which becomes 
\[
\APPT \subseteq\LS_3
\subseteq\ARED
\subseteq\ASN_{2},
\]
for $d=3$, and 
$\APPT = \ARED = \ASN_1 = \text{ASEP}$ for $d=2$, where ASEP is the set of all absolutely separable states.

To show that every inclusion in this chain is strict, we use a slight modification of the pseudo-pure states $\rho_\alpha = (I + \alpha P)/(d^2 + \alpha)$ in Refs.~\cite{Jivulescu2015,AbellanetVidalEtAl2026}.

\begin{Example}
For $d\ge4$, let $P,Q$ be rank-one
projections with $PQ=0$.
Define
\[
\rho_{\alpha,t}
:=
\frac{I+\alpha P-(1-t)Q}{d^2+\alpha-1+t},
\quad
\alpha\ge0,\; 0<t\le1.
\]
Its spectrum is proportional to $(1+\alpha,1,\dots,1,t)$
and therefore
$\rho_{\alpha,t}\in\LS_p$ if and only if $\alpha\le p-2+t.$

Fix $t=t_d:=(d-1)/d$. Then
$\rho_{\alpha,t_d}\in\LS_p$ if and only if $\alpha\le p-1-\frac1d.$
For $\alpha=2-1/d$, the state belongs to $\LS_3$.
On the other hand, the APPT necessary condition
\[
\lambda_1\le\lambda_{d^2-1}
+2\sqrt{\lambda_{d^2-2}\lambda_{d^2}}
\]
from Ref.~\cite{Hildebrand2007PPT}
reduces to $\alpha\le2\sqrt{t_d}$.
Since $2-\frac1d=1+t_d>2\sqrt{t_d}$, 
we have $\rho_{2-1/d,t_d}\in\LS_3\setminus\APPT.$

Taking $\alpha=d-1-1/d$ gives $\rho_{d-1-1/d,t_d}\in\LS_d\setminus\LS_3.$

To treat $\ARED$, observe that
\[
I+\alpha P-(1-t)Q
=
t\left(I+\frac{\alpha}{t}P\right)+(1-t)(I-Q).
\]
Hence $\rho_{\alpha,t}$ is a convex combination of a
pseudo-pure state with parameter $\alpha/t$ and the state
$(I-Q)/(d^2-1)$.
Since the former belongs to $\ARED$ if
$\alpha/t\le d$ by Ref.~\cite{Jivulescu2015}, while
$(I-Q)/(d^2-1)\in\LS_d\subseteq\ARED$, the convexity of $\ARED$
gives $\rho_{\alpha,t}\in\ARED$ for $\alpha\le dt$.

Therefore, for $\alpha=d-1-1/(2d)$,
$
\rho_{d-1-1/(2d),t_d}
\in
\ARED\setminus\LS_d.
$

Next take $\alpha=d$.
Since $d\le2d-2-\frac1d$, we have $\rho_{d,t_d}\in\LS_{2d-1}.$
Since $\bra{\Phi}\rho\ket{\Phi}\le\bra{\Phi}\rho_A\otimes I\ket{\Phi}=1/d$ for every maximally entangled $\ket{\Phi}$, every ARED spectrum has largest eigenvalue at most $1/d$, whereas
\[
\lambda_{\max}(\rho_{d,t_d})
=
\frac{d+1}{d^2+d-1/d}
>
\frac1d.
\]
Hence
$
\rho_{d,t_d}
\in
\LS_{2d-1}\setminus\ARED$.

For $\alpha=2d-2$, we have $2d-2>2d-2-\frac1d$, so $\rho_{2d-2,t_d}\notin\LS_{2d-1}$.
Since $\lfloor d^2/2\rfloor\ge2d$ for $d\ge4$,
$
2d-2
\le
\left\lfloor\frac{d^2}{2}\right\rfloor-1-\frac1d,
$
and therefore
$
\rho_{2d-2,t_d}
\in
\LS_{\lfloor d^2/2\rfloor}\setminus\LS_{2d-1}.
$

Finally, using the same convex decomposition and the pseudo-pure
criterion $\rho_{\beta,1}\in\ASN_{d-1}$ whenever $\beta\le2(d^2-d-1)$ from Ref.~\cite{AbellanetVidalEtAl2026}, the convexity of
$\ASN_{d-1}$ gives $\rho_{\alpha,t}\in\ASN_{d-1}$ for $\alpha\le2t(d^2-d-1)$.
Take $\alpha=\lfloor d^2/2\rfloor-1$.
Then
$
\rho_{\lfloor d^2/2\rfloor-1,t_d}
\notin
\LS_{\lfloor d^2/2\rfloor},
$
while
$
\left\lfloor\frac{d^2}{2}\right\rfloor-1
<
2t_d(d^2-d-1)
$ for every $d\ge4$.
Thus
$
\rho_{\lfloor d^2/2\rfloor-1,t_d}
\in
\ASN_{d-1}\setminus\LS_{\lfloor d^2/2\rfloor}.
$

In summary, the strictness of all inclusions is witnessed by
\[
\begin{aligned}
\rho_{\,2-1/d,t_d}
&\in \LS_3\setminus\APPT,\\
\rho_{\,d-1-1/d,t_d}
&\in \LS_d\setminus\LS_3,\\
\rho_{\,d-1-1/(2d),t_d}
&\in \ARED\setminus\LS_d,\\
\rho_{\,d,t_d}
&\in \LS_{2d-1}\setminus\ARED,\\
\rho_{\,2d-2,t_d}
&\in \LS_{\lfloor d^2/2\rfloor}\setminus\LS_{2d-1},\\
\rho_{\,\lfloor d^2/2\rfloor-1,t_d}
&\in \ASN_{d-1}\setminus\LS_{\lfloor d^2/2\rfloor}.
\end{aligned}
\]
\end{Example}

\subsection{The global unitary requirement cannot be removed}

Ordinary reduction positivity does not exclude maximal Schmidt number, even when the reduction criterion is imposed on both subsystems.

\begin{Thm}\label{thm:maxred}
For every $d\ge3$, there exists a bipartite state $\rho\in\M_d\otimes\M_d$ such that
\[
\rho\in\RED_A\cap\RED_B
\text{ and }
\SN(\rho)=d.
\]
Equivalently, $\RED_A\cap\RED_B
\not\subseteq
\SN_{d-1}.$
\end{Thm}

The explicit construction and the proof are given in Appendix~\ref{app:maximal-SN-RED}.
Thus, in contrast to $\ARED\subseteq\ASN_{d-1}$, ordinary reduction positivity is compatible with maximal Schmidt number in every dimension $d\ge3$.

\section{Discussion}\label{sec:discussion}

We studied spectral constraints on $(d-1)$-block positive operators in order to obtain spectrum-only obstructions to maximal Schmidt number. 
Our main result is that $\LS_p\subseteq\ASN_{d-1}$ for every $p\le\lfloor d^2/2\rfloor$; that is, whenever the largest eigenvalue of a state does not exceed the sum of its $p$ smallest eigenvalues, the spectrum alone guarantees that its Schmidt number is not maximal.
This also yields a simple sufficient condition involving only the largest eigenvalue (Corollary~\ref{cor:lmax}).

The spectral inclusion has a direct application to $\ARED$. 
For $d\ge4$, the known inclusion $\ARED\subseteq\LS_{2d-1}$ lies inside the range of the general theorem. 
The qutrit case falls just outside this range and requires additional spectral information. 
The complex symmetric bilinear form $\beta$ defined in Eq.~\eqref{eq:takagi-setup} provides the paired inequalities in Eq.~\eqref{eq:qutrit-pairs}, which, combined with the spectral condition in Eq.~\eqref{eq:qutrit-ared-min} satisfied by qutrit ARED states, complete the proof.
Consequently, $\ARED\subseteq\ASN_{d-1}$ holds in every dimension,  extending the corresponding result for APPT states~\cite{HuberLamiLancienMullerHermes2018}.

The contrast with ordinary reduction positivity is sharp. 
For every $d\ge3$, Theorem~\ref{thm:maxred} gives a state of maximal Schmidt number satisfying the reduction criterion on both subsystems. 
Thus the spectral obstruction to maximal Schmidt number is a genuinely absolute phenomenon rather than a consequence of the reduction criterion alone.

Several questions remain. 
Theorem~\ref{thm:main} provides a sufficient range but does not determine the largest $p$ for which $\LS_p\subseteq\ASN_{d-1}$, whereas it can be shown that $\LS_p\not\subseteq\ASN_{d-1}$ if $p\ge \frac{d^2+d}{2}$~\cite{JSpre}.
The qutrit case is of particular interest: Theorem~\ref{thm:main} covers $\LS_4$, while $\ARED \subseteq \LS_5$, so it is natural to ask whether $\LS_5 \subseteq \ASN_2$.
More generally, it would be interesting to determine analogous $\LS_p$ ranges for $\ASN_k$ with $k<d-1$, and to understand whether ordinary PPT states can have maximal Schmidt number. 
Although PPT states with large Schmidt number are known~\cite{HuberLamiLancienMullerHermes2018,johnston2026ppt}, this question remains open.

\noindent{\bf AI Statement.}
ChatGPT (OpenAI), including GPT-5.5 and GPT-5.6 Sol, was used for language editing, LaTeX assistance, and mathematical discussion related to the proofs of Lemmas~\ref{lem:normal-subspace} and~\ref{lem:takagi-lower} and the construction of states with maximal Schmidt number satisfying the reduction criterion on both subsystems. The authors independently verified the mathematical suggestions and take full responsibility for the final manuscript.

\begin{acknowledgments}
This research was supported by the Institute of Information \& Communications Technology Planning \& Evaluation (IITP) grant funded by the Ministry of Science and ICT (MSIT) (No. RS-2025-02304540).
S.L. acknowledges support from the National Research Foundation of Korea (NRF) of Korea grants funded by the MSIT  (No. RS-2024-00432214 and No. RS-2022-NR068791) and Creation of the Quantum Information Science R\&D Ecosystem (No. RS-2023-NR068116) through the NRF funded by the MSIT.
\end{acknowledgments}

\bibliographystyle{apsrev4-2}
\bibliography{aredsnfs}

\appendix


\section{Proof of Theorem~\ref{thm:bp-inertia}}\label{app:bp-inertia}

In this appendix, we prove Theorem~\ref{thm:bp-inertia}.
Throughout this appendix, $W\in\BP_{d-1}$ is not positive.

\begin{proof}
Assume that $\vket{Y}\in\ker W$.
We show that $Y=0$.

First we let $Z \in GL_d$ satisfy $\langle\!\langle Z|W|Z\rangle\!\rangle < 0$, and let $t$ be an eigenvalue of $Z^{-1}Y$. Since $\det(tZ - Y) = \det(Z)\det(tI - Z^{-1}Y) = 0$, the matrix $tZ - Y$ is singular, so $\mathrm{SR}(|tZ - Y\rangle\!\rangle) \le d-1$. As $W|Y\rangle\!\rangle = 0$, $(d-1)$-block positivity gives
$$0 \le \langle\!\langle tZ - Y|W|tZ - Y\rangle\!\rangle = |t|^2 \langle\!\langle Z|W|Z\rangle\!\rangle,$$
and the right-hand side is negative unless $t = 0$. Hence every eigenvalue of $Z^{-1}Y$ is zero, and in particular $\mathrm{Tr}(Z^{-1}Y) = 0$.

Now we let $|X\rangle\!\rangle$ be the negative eigenvector of $W$. Then $\mathrm{rank}\,X = d$ and $\langle\!\langle X|W|X\rangle\!\rangle < 0$, so the preceding observation gives $\mathrm{Tr}(X^{-1}Y) = 0$. For $\epsilon > 0$ define $Z_\epsilon := (X^{-1} + \epsilon Y^\dagger)^{-1}$. For sufficiently small $\epsilon$, $Z_\epsilon$ is well defined and $Z_\epsilon \to X$, so by continuity $\langle\!\langle Z_\epsilon|W|Z_\epsilon\rangle\!\rangle < 0$. Applying the preceding observation to $Z_\epsilon$,
$$0 = \mathrm{Tr}(Z_\epsilon^{-1}Y) = \mathrm{Tr}(X^{-1}Y) + \epsilon\,\mathrm{Tr}(Y^\dagger Y) = \epsilon\,\mathrm{Tr}(Y^\dagger Y),$$
which implies $Y = 0$. Hence $\ker W = \{0\}$, and so $W$ is invertible.
\end{proof}

\section{Proof of Lemma~\ref{lem:normal-subspace}}\label{app:normal}

In this appendix, we prove Lemma~\ref{lem:normal-subspace}.

Regard $\M_d$ as a real vector space of dimension $2d^2$ and define
\[
\mathcal S:=\{H+itI_d:H=H^\dagger,\ t\in\mathbb R\}.
\]
Every matrix in $\mathcal S$ is normal, and $\dim_{\mathbb R}\mathcal S=d^2+1$. 
If $m=\dim_{\C}V$, then $\dim_{\mathbb R}V=2m$.
Therefore
\[
\dim_{\mathbb R}(V\cap\mathcal S)
\ge2m+(d^2+1)-2d^2
=2m-d^2+1.
\]
For $m\ge\lceil d^2/2\rceil$, the right-hand side is at least one. 
Hence $V\cap\mathcal S$ contains a nonzero normal matrix of the stated form $H+cI_d$.

\section{The matrix representation of \texorpdfstring{$\beta$}{beta}}
\label{appen:S_beta}

We collect here the coordinate facts used in Lemma~\ref{lem:takagi-lower}.
Let $\mathcal F=\{F_1,\ldots,F_{d^2-1}\}$ and
$\mathcal G=\{G_1,\ldots,G_{d^2-1}\}$
be two Hilbert--Schmidt orthonormal bases of $A^\perp$.
Define $C_{ij}=\Tr(F_i^\dagger G_j)=\vbra{F_i}G_j\rangle\!\rangle$.
Then
\[
(C^\dagger C)_{ij}
=\sum_k C_{ki}^*\,C_{kj}
=\sum_k \vbra{G_i}F_k\rangle\!\rangle\vbra{F_k}G_j\rangle\!\rangle
=\vbra{G_i}G_j\rangle\!\rangle
=\delta_{ij},
\]
so $C^\dagger C=I$, $C=(C_{ij})$ is unitary.
Also we have $G_j=\sum_{i=1}^{d^2-1} C_{ij}F_i$.
Consequently, $S_\beta^{\mathcal G}=C^T S_\beta^{\mathcal F}C.$
Thus the coordinate matrices of $\beta$ associated with two Hilbert--Schmidt orthonormal bases are related by unitary congruence and have the same singular values.

We also justify the reduction to the case where $A$ is diagonal.
Let $A=U\Sigma V^\dagger$ be a singular value decomposition and define $\mathcal T(Y)=U^\dagger YV$.
Then $\mathcal T$ is a Hilbert--Schmidt isometry from $A^\perp$ onto $\Sigma^\perp$.
Moreover,
\[
\beta(Y,Z)
=
\frac1d\Tr\left(
\Sigma^{-1}\mathcal T(Y)\Sigma^{-1}\mathcal T(Z)
\right).
\]
Hence an orthonormal basis of $A^\perp$ may be transported to an orthonormal basis of $\Sigma^\perp$ without changing the coordinate matrix of the bilinear form.
Therefore, in the proof of Lemma~\ref{lem:takagi-lower}, we may assume
\[
A=\Sigma=\operatorname{diag}(s_1,\ldots,s_d),
\]
where $s_i>0$ and $\sum_i s_i^2=1$.

\section{Proof details for Lemma~\ref{lem:takagi-lower}}
\label{appen:pf-takagi-lower}

We only need to justify the estimate on the diagonal part.
By Appendix~\ref{appen:S_beta}, we may assume that $A=\operatorname{diag}(s_1,\ldots,s_d)>0$ with $\sum_{i=1}^d s_i^2=1$.
Identify diagonal matrices with vectors in $\C^d$ and let $\ket{s}=(s_1,\ldots,s_d)$.
Choose a real orthonormal basis of $|s\rangle^\perp = \{|x\rangle : \langle x|s\rangle = 0\}$, and let $E_1, \dots, E_{d-1}$ be the corresponding diagonal matrices. They form a Hilbert–Schmidt orthonormal basis of $\mathcal{D}$. Let $S_{\beta,\mathcal{D}} = [\beta(E_k, E_l)]_{k,l=1}^{d-1}$ be the matrix of the restriction of $\beta$ to $\mathcal{D}$ in this basis. Since each $E_k$ is real and diagonal, $\beta(E_k, E_l) = \frac{1}{d}\sum_i (E_k)_{ii}(E_l)_{ii}/s_i^2$ is real, so $S_{\beta,\mathcal{D}}$ is real symmetric and its singular values are the absolute values of its eigenvalues.

For $D = \sum_k x_k E_k \in \mathcal{D}$ with $x \in \mathbb{C}^{d-1}$, we have $D^* = \sum_k x_k^* E_k$, hence

$$\beta(D^*, D) = \sum_{k,l} x_k^* x_l\, \beta(E_k, E_l) = x^\dagger S_{\beta,\mathcal{D}}\, x, \qquad \|D\|_2 = \|x\|.$$

Therefore every eigenvalue, and hence every singular value, of $S_{\beta,\mathcal{D}}$ is at least $2/d$ as soon as $\beta(D^*, D) \ge \frac{2}{d}\|D\|_2^2$ for all $D \in \mathcal{D}$. Writing $D = \mathrm{diag}(z_1, \dots, z_d)$, we have $\beta(D^*, D) = \frac{1}{d}\sum_i |z_i|^2/s_i^2$, so it suffices to prove 
\begin{equation}\label{eq:appendix-diagonal-bound}
\sum_i\frac{|z_i|^2}{s_i^2}\ge2\sum_i|z_i|^2
\end{equation}
under the constraint $\sum_i s_i z_i=0$.

Set $t_i=s_i^2$ and $z_i=s_i x_i$.
Then $\sum_i t_i=1$ and $\sum_i t_i x_i=0$ and Eq.~\eqref{eq:appendix-diagonal-bound} becomes
\[
\sum_i|x_i|^2\ge2\sum_i t_i|x_i|^2.
\]
If $t_i\le1/2$ for every $i$, the claim is immediate.

Otherwise, after relabeling, let $t_1=t>\frac12$ and $T=1-t$.
Then $t_i\le T$ for $i\ge2$.
Put
\[
R=\sum_{i=2}^d t_i|x_i|^2,
\;
S=\sum_{i=2}^d|x_i|^2.
\]
Since $\sum_i t_i x_i=0$, by the Cauchy--Schwarz inequality,
\[
t^2|x_1|^2
=\left|\sum_{i=2}^d t_ix_i\right|^2
\le TR,
\]
while $R\le TS$.
Therefore
\[
\begin{aligned}
\sum_i|x_i|^2-2\sum_i t_i|x_i|^2
&=S-2R-(2t-1)|x_1|^2\\
&\ge\frac{R}{T}-2R-\frac{(2t-1)T}{t^2}R\\
&=\frac{(2t-1)^2}{t^2T}R
\ge0.
\end{aligned}
\]
Hence $S_{\beta,\mathcal{D}} \ge \frac{2}{d} I$, as required.

\section{Two-sided reduction-positive states with maximal Schmidt number}\label{app:maximal-SN-RED}

In this appendix, we construct states in $\RED_A\cap\RED_B$ with maximal Schmidt number in three cases: odd $d$, even $d\ge6$ and $d=4$.

\subsection{Odd case}

Let $d=2k+1$ with $k\ge1$.
For $j=1,\ldots,k$, define
\[
\begin{aligned}
N_{2j-1}&=\ket{2j-1}\bra{2k+1}+\ket{2k+1}\bra{2j},\\
N_{2j}&=\ket{2j}\bra{2k+1}-\ket{2k+1}\bra{2j-1}.
\end{aligned}
\]
Let $w=\frac{2k-1+\sqrt{(2k-1)^2+4}}{2}$, so that $w^2-(2k-1)w-1=0$. 
Define the state $\rho^{(2k+1)}\in\M_d\otimes \M_d$
\[
\rho^{(2k+1)}
=
\frac{1}{2k+1+4kw}
\left(
\vket{I_{2k+1}}\vbra{I_{2k+1}}
+w\sum_{j=1}^{2k}\vket{N_j}\vbra{N_j}
\right).
\]

Every vector in the range of $\rho^{(2k+1)}$ can be written as $\vket{X}=a\vket{I_{2k+1}}+\sum_{j=1}^{2k}x_j\vket{N_j}$, so its corresponding matrix $X$ is $X=aI_{2k+1}+\sum_{j=1}^{2k}x_jN_j$.

Note that
\[
\begin{aligned}
\sum_{j=1}^{2k}x_jN_j 
&=\sum_{j=1}^{k}(x_{2j-1}N_{2j-1} + x_{2j}N_{2j})\\
&=\sum_{j=1}^{k}x_{2j-1}(\ket{2j-1}\bra{2k+1}+ \ket{2k+1}\bra{2j}) \\ &\qquad \qquad+ x_{2j}(\ket{2j}\bra{2k+1}-\ket{2k+1}\bra{2j-1})\\
&=\sum_{j=1}^{2k}x_j\ket{j}\bra{2k+1}\\&\qquad\qquad+\ket{2k+1}(\sum_{j=1}^{k}x_{2j-1}\bra{2j}-x_{2j}\bra{2j-1}).
\end{aligned}
\]
Put the unnormalized state $\ket{x}=\sum_{j=1}^{2k}x_j\ket{j}$ and $\ket{\tilde{x}}=\sum_{j=1}^{k}x_{2j-1}^*\ket{2j}-x_{2j}^*\ket{2j-1}$.
Then we have
$\sum_{j=1}^{2k}x_jN_j=\ket{x}\bra{2k+1}+\ket{2k+1}\bra{\tilde{x}}$. 
Since $\bra{\tilde{x}}\!x\rangle=0$, we have $(\sum_{j=1}^{2k}x_jN_j)^3=0$.
Since $N = \sum_j x_j N_j$ is nilpotent, all eigenvalues of $aI + N$ equal $a$.
Therefore
\[
\det(aI_{2k+1}+\sum_{j=1}^{2k}x_jN_j)=a^{2k+1}.
\]
Hence X is singular if and only if $a = 0$.
It follows that every vector in the range of $\rho^{(2k+1)}$ having Schmidt rank at most $2k$ belongs to  $\operatorname{span}\{\vket{N_1},\ldots,\vket{N_{2k}}\}$, which is a proper subspace of the range because $\vket{I_{2k+1}}$ also belongs to the range. 

Now suppose $\SN(\rho^{(2k+1)})\le2k$.
By definition for Schmidt number, there exists a decomposition of $\rho^{(2k+1)}=\sum_jp_j\vket{X_j}\vbra{X_j}$ with all
vectors $\vket{X_j}$ of Schmidt rank at most $2k$.
Then $\operatorname{span}\{\vket{X_j}\}$ is the range of $\rho^{(2k+1)}$, which is impossible by the preceding observation.
Hence $\SN(\rho^{(2k+1)})=2k+1$.

The definition of the $N_j$ implies
\[
\sum_{j=1}^{2k}N_jN_j^\dagger
=
\sum_{j=1}^{2k}N_j^\dagger N_j
=
\operatorname{diag}(1,\ldots,1,2k).
\]
Hence $\rho_A^{(2k+1)}=\rho_B^{(2k+1)}=\frac{1}{(2k+1+4kw)}D$, where
\[
D=\operatorname{diag}(1+w,\ldots,1+w,1+2kw).
\]

Consider $|u_0\rangle = (D^{-1/2}\otimes I)|I\rangle\!\rangle$ and $|u_j\rangle = \sqrt{w}\,(D^{-1/2}\otimes I)|N_j\rangle\!\rangle$ for $j = 1, \dots, 2k$. The inner products $\langle u_i|u_j\rangle$ are proportional to $\mathrm{Tr}(D^{-1}N_jN_i^\dagger)$ or $\mathrm{Tr}(D^{-1}N_j)$, which vanish for $i \ne j$ because the corresponding matrices have zero diagonal. Moreover, $\langle u_0|u_0\rangle = \mathrm{Tr}(D^{-1}) = \frac{2k}{1+w} + \frac{1}{1+2kw}$ and $\langle u_j|u_j\rangle = w\big(\frac{1}{1+w} + \frac{1}{1+2kw}\big)$, and both equal $1$ by $w^2 - (2k-1)w - 1 = 0$. Hence $\{|u_j\rangle\}_{j=0}^{2k}$ is an orthonormal set.

Therefore
\[
((\rho_A^{(2k+1)})^{-1/2}\otimes I_B)
\rho^{(2k+1)}
((\rho_A^{(2k+1)})^{-1/2}\otimes I_B) 
=\sum_{j=0}^{2k}\ket{u_j}\bra{u_j}
\le I.
\]
This proves $(\mathrm{id} \otimes R)(\rho^{(2k+1)}) = \rho_A^{(2k+1)} \otimes I_d - \rho^{(2k+1)}\ge 0$.
Using $\sum_jN_j^\dagger N_j=\operatorname{diag}(1,\ldots,1,2k)$ in the same way gives $(R \otimes \mathrm{id})(\rho^{(2k+1)}) = I_d \otimes \rho_B^{(2k+1)} - \rho^{(2k+1)}\ge 0$.

\subsection{Even case for $d\ge 6$}\label{subsec:6}

Let $d=2k+2$ with $k\ge2$. 
Define
\[
\begin{aligned}
C_0={}&3\ket1\bra1-\frac23\ket1\bra2+\frac23\ket2\bra1
+\frac43\ket2\bra2+\sum_{j=3}^{2k+2}\ket j\bra j,\\
C_1={}&\ket1\bra2-\ket2\bra1+2\ket2\bra2.
\end{aligned}
\]
For $j=1,\ldots,k$, define
\[
\begin{aligned}
N_{2j-1}&=\ket{2j+1}\bra2+\ket2\bra{2j+2},\\
N_{2j}&=\ket{2j+2}\bra2-\ket2\bra{2j+1}.
\end{aligned}
\]
Let $V=[\vket{C_0},\vket{C_1},\vket{N_1},\ldots,\vket{N_{2k}}]$ and for real $r,b,s$, 
$K=K(r,b,s)
=
\begin{pmatrix}
1&r\\
r&b
\end{pmatrix}
\oplus sI_{2k}.$

Consider the unnormalized operator $\rho'= VKV^\dagger$.
Every vector in the range of $\rho'$ belongs to the range of $V$ and hence can be written as $\vket{X}=a\vket{C_0}+c\vket{C_1}+\sum_{l=1}^{2k}y_l\vket{N_l}$.
Its corresponding matrix $X$ is $X=aC_0+cC_1+\sum_{l=1}^{2k}y_l N_l$.
With respect to the decomposition $\C^{2k+2} = \C^2\oplus \C^{2k}$, $X$ has the block form
\[
X = \begin{bmatrix}
3a & c-\frac{2a}{3} & 0_{1\times 2k}\\
-c+\frac{2a}{3} & \frac{4a}{3}+2c & \sum_{l=1}^k (y_{2l-1}\bra{2l} -y_{2l}\bra{2l-1}) \\
0_{2k\times 1} & \sum_{l=1}^{2k}y_l\ket{l} & aI_{2k}
\end{bmatrix}.
\]
For $a \ne 0$, the Schur complement of the block $aI_{2k}$ differs from the upper-left $2\times2$ block only in its $(2,2)$ entry, by $-a^{-1}\big(\sum_l (y_{2l-1}\langle 2l| - y_{2l}\langle 2l-1|)\big)\big(\sum_l y_l|l\rangle\big) = 0$, by the same cancellation as in the odd case. Hence $\det X = a^{2k}\det\begin{pmatrix} 3a & c - 2a/3 \\ -c + 2a/3 & 4a/3 + 2c\end{pmatrix} = \frac{a^{2k}}{9}(3c+4a)(3c+10a)$.
Thus $X$ is singular if and only if $a=0$, $c=-4a/3$, or $c=-10a/3$.
Define $q(a,c)=-2|a|^2/3-(a^*c+ac^*)/4$. 
In these three cases, $q(a,c)$ equals $0$, $0$, and $|a|^2$, respectively.

Assume $K>0$ and $r>-4/3$, and suppose that$VKV^\dagger = \sum_\alpha |X_\alpha\rangle\!\rangle\langle\!\langle X_\alpha|$ with every $X_\alpha$ singular. Each $|X_\alpha\rangle\!\rangle$ lies in the range of $V$, so $|X_\alpha\rangle\!\rangle = V|z_\alpha\rangle$ with $|z_\alpha\rangle = (a_\alpha, c_\alpha, y_{\alpha,1}, \dots, y_{\alpha,2k})$. Since $V$ has linearly independent columns, $K = \sum_\alpha |z_\alpha\rangle\langle z_\alpha|$. As each $X_\alpha$ is singular, $q(a_\alpha, c_\alpha) \ge 0$, and therefore 
\[
\begin{aligned}
0\le\sum_\alpha q(a_\alpha,c_\alpha)
&= -\frac23\sum_\alpha|a_\alpha|^2 -\frac14\sum_\alpha(a_{\alpha*} c_\alpha+a_\alpha c_{\alpha*})\\
&=
-\frac23-\frac{r}{2}<0,    
\end{aligned}
\]
which is a contradiction.
Hence $\SN(\rho')=2k+2$ whenever $K>0$ and $r>-\frac43$.

Since the reduction condition is invariant under multiplication by a positive scalar, if $\rho'=VKV^\dagger$ satisfies the reduction criterion on both sides, then its normalized state also satisfies it.
Its marginals are
\[
\rho_A'
=
\begin{pmatrix}
\frac{85}{9}+b-\frac43r&\frac{10}{9}+2b-3r\\
\frac{10}{9}+2b-3r&\frac{20}{9}+5b+4r+2ks
\end{pmatrix}
\oplus(1+s)I_{2k},
\]
and
\[
\rho_B'
=
\begin{pmatrix}
\frac{85}{9}+b-\frac43r&-(\frac{10}{9}+2b-3r)\\
-(\frac{10}{9}+2b-3r)&\frac{20}{9}+5b+4r+2ks
\end{pmatrix}
\oplus(1+s)I_{2k}.
\]

We now choose $r,b,s$ so that $\rho'$ satisfies the reduction
criterion on both sides.
For $-\frac43<r<\frac23$, set
\[
b(r)
=
\frac{36r^2-306r-1270}{9(3r-2)}
\]
and let $s(r)>0$ be the positive solution of
\[
s(r)(s(r)+1)
=
\frac{(3r+160)(27r^3-54r^2+306r+1270)}
{81(3r-2)^2}.
\]
The right-hand side is positive, so $s(r)$ is well-defined: indeed, $3r+160>0$ and $f(r)=27r^3-54r^2+306r+1270>0$, since $f(-4/3)=702$ and $f'(r)>0$. 
Define
\begin{equation}\label{eq:existence}
F_k(r)
=
2k\,s(r)
-
\frac{12(3r+7)(25+5r-3r^2)}
{(3r-2)^2}.
\end{equation}
The function $F_k$ is continuous on $[-4/3, 2/3)$, with $F_k(-\frac43)=\frac{k}{3}\left(\sqrt{1361}-3\right)-13>0$ for every $k\ge2$ and $\lim_{r\to(2/3)^-}F_k(r)=-\infty$.
Therefore, there exists
$
r_k\in\left(-\frac43,\frac23\right)
$
such that
\[
2k\,s(r_k)
=
\frac{12(3r_k+7)(25+5r_k-3r_k^2)}
{(3r_k-2)^2}.
\]

Set $b_k=b(r_k)$, $s_k=s(r_k)$ and $K_k = K(r_k,b_k,s_k)$.
Moreover,
\[
b_k-r_k^2
=
\frac{27r_k^3-54r_k^2+306r_k+1270}
{9(2-3r_k)}
>0,
\]
so \(K_k>0\).
Define $\rho'=VK_kV^\dagger$ and $\rho^{(2k+2)} = \frac{VK_kV^\dagger}{\Tr(VK_kV^\dagger)} $.

Using the equations defining $r_k,s_k,b_k$, one obtains
\[
V^\dagger(\rho_A'^{-1}\otimes I_{2k+2})V
=
V^\dagger(I_{2k+2}\otimes\rho_B'^{-1})V
=
K_k^{-1}.
\]
Since $K_k^{1/2}V^\dagger(\rho_A'^{-1/2}\otimes I_{2k+2})(\rho_A'^{-1/2}\otimes I_{2k+2})VK_k^{1/2}=I$,
\[
(\rho_A'^{-1/2}\otimes I_{2k+2})
\rho'
(\rho_A'^{-1/2}\otimes I_{2k+2})
\le I.
\]

Thus $(\mathrm{id} \otimes R)(\rho^{(2k+2)}) = \rho_A^{(2k+2)} \otimes I_d - \rho^{(2k+2)}\ge 0$.
Similarly, we have $(R \otimes \mathrm{id})(\rho^{(2k+2)}) = I_d \otimes \rho_B^{(2k+2)} - \rho^{(2k+2)}\ge 0$.

\subsection{Even case for $d=4$}

In subsection~\ref{subsec:6}, the existence argument above fails for $k=1$, since $F_1(-4/3)<0$.
Therefore we treat $d=4$ separately.

Let
\[
\begin{aligned}
A_0&=\frac12I_4,&
A_1&=\frac{\ket3\bra1-\ket4\bra2}{\sqrt2},&
A_2&=\frac{\ket1\bra2+\ket3\bra4}{\sqrt2},\\
A_3&=\ket2\bra3,&
A_4&=\ket3\bra2.
\end{aligned}
\]
Set $V=
[\vket{A_0},\vket{A_1},\vket{A_2},\vket{A_3},\vket{A_4}]$ and let
\[
G
=
\frac1{1000}
\begin{pmatrix}
204&0&0&-35&-60\\
0&210&0&0&0\\
0&0&210&0&0\\
-35&0&0&295&-94\\
-60&0&0&-94&81
\end{pmatrix}.
\]
After a permutation of coordinates for $G$, it can be written as
\[
210\oplus210\oplus 
\begin{pmatrix}
204&-35&-60\\
-35&295&-94\\
-60&-94&81
\end{pmatrix}.
\]
The matrices $A_0,\ldots,A_4$ are Hilbert--Schmidt orthonormal, and $\Tr G=1$. 
Direct calculation shows that \(G>0\), so that $\rho^{(4)}:=VGV^\dagger=\sum_{i,j=0}^4G_{ij}\vket{A_i}\vbra{A_j}$ is a state.

Every vector in the range of $\rho^{(4)}$ has the form $\vket{X(z)}=\sum_{i=0}^4z_i\vket{A_i}$, corresponding to the matrix $X(z)=\sum_{i=0}^4z_iA_i$. 
Direct calculation gives $\det X(z)=\frac{z_0^2}{4}(\frac{z_0^2}{4}-z_3z_4)$.

Define
\[
q(z_0,z_3,z_4)
=
|z_3+3z_4|^2
-\sqrt3\,(z_0^*z_3+z_0z_3^*)
+3\sqrt3\,(z_0^*z_4+z_0z_4^*).
\]
If $z_0=0$, then $q(z_0,z_3,z_4)=|z_3+3z_4|^2\ge0$. 
If $z_0^2=4z_3z_4$ and $z_0\neq0$, choose $x,y\in\C$ such that $z_3=x^2$, $z_4=y^2$, and $z_0=2xy$.
Then
\[
q(z_0,z_3,z_4)
=
\left(|x|^2-3|y|^2-\sqrt3(x^*y+xy^*)\right)^2\ge0.
\]
Thus $q(z_0,z_3,z_4)\ge0$ whenever $X(z)$ is singular.

Suppose $\SN(\rho^{(4)})\le3$.
Then there exists a decomposition $\rho^{(4)}=\sum_\alpha\vket{X_\alpha}\vbra{X_\alpha}$ with every $X_\alpha$ singular. 
Using the same argument as in subsection~\ref{subsec:6}, 
$G=\sum_\alpha \ket{z_\alpha}\bra{z_\alpha}$ because the $A_i$ are linearly independent.
Hence
\[
\begin{aligned}
\sum_\alpha q(z_\alpha)
={}&G_{33}+9G_{44}+6G_{34}
-2\sqrt3G_{03}
+6\sqrt3G_{04}\\
={}&\frac{46-29\sqrt3}{100}<0,
\end{aligned}
\]
contradicting $q(z_\alpha)\ge0$. 
Therefore $\SN(\rho^{(4)})=4$.

Finally, direct calculation yields
\[
2000\rho_A^{(4)}
=
\begin{pmatrix}
312&0&0&0\\
0&692&-95&0\\
0&-95&684&0\\
0&0&0&312
\end{pmatrix},
\]
\[
2000\rho_B^{(4)}
=
\begin{pmatrix}
312&0&0&0\\
0&684&-95&0\\
0&-95&692&0\\
0&0&0&312
\end{pmatrix}.   
\]
Substituting the $A_i$ and $G$ yields $\rho_A^{(4)}\otimes I_4\ge \rho^{(4)}$ and $I_4\otimes\rho_B^{(4)}\ge\rho^{(4)}$.

\end{document}